\documentclass[11pt]{article}
\usepackage{graphicx}
\usepackage{amsmath,amssymb,amsthm,enumerate}
\usepackage{amsfonts,mathrsfs}
\usepackage{bm,indentfirst}
\usepackage{color}

\graphicspath{{fig/}}

\allowdisplaybreaks[4]

\begin{document}
\renewcommand{\d}{\mathrm{d}}
\newtheorem{theorem}{Theorem}[section]
\numberwithin{equation}{section}
\renewcommand{\figurename}{Fig.}

\title{A tensor model for nematic phases of bent-core molecules based on molecular theory}
\author{Jie Xu$^1$\footnote{Current address: Department of Mathematics, Purdue University, West Lafayette, IN 47907, USA}, Fangfu Ye$^{2,3}$ \& Pingwen Zhang$^1$\footnote{Corresponding author}\\[2mm]
{\small $^1$LMAM \& School of Mathematical Sciences, Peking University, 
  Beijing 100871, China}\\[1mm]
{\small $^2$Beijing National Laboratory for Condensed Matter Physics, }\\
{\small Institute of Physics, Chinese Academy of Sciences, Beijing 100190, China}\\[1mm]
{\small $^3$School of Physical Sciences, University of Chinese Academy of Sciences,
  Beijing 100094, China}\\[3mm]
{\small E-mail: rxj\_2004@126.com,\, fye@iphy.ac.cn,\, pzhang@pku.edu.cn}\\
}
\date{\today}
\maketitle
\begin{abstract}
We construct a tensor model for nematic phases of bent-core molecules from molecular theory. 
The form of free energy is determined by molecular symmetry, 
which includes the couplings and derivatives of a vector and two second-order tensors, 
with the coefficients determined by molecular parameters. 
We use the model to study the nematic phases resulted from the hard-core potential. 
Unlike most macroscopic models, we are able to obtain the phase diagram about the molecular parameters, but not merely some phenomenological coefficients. 
The tensor model is applicable to other molecules with the same symmetry, 
which we demonstrate by studying the phase diagram of star molecules. 

\vspace{12pt}
\end{abstract}

\section{Introduction}
The ability to show complex orientational order has drawn much attention 
of the liquid crystal community to bent-core molecules. 
This feature originates from the $C_{2v}$\footnote{This is the Sch\"onflies notation.} molecular symmetry that breaks the 
axisymmetry of a rod-like molecule. 
The polar and biaxial order is notable in layer or columnar structures \cite{Smectic,JJAP}. 
The homogeneous biaxial nematic phase is also observed \cite{BiExp_prl2004,BiExp_prl2004_2} spontaneously formed by bent-core molecules without imposing external forces. 
Moreover, bent-core molecules are able to exhibit modulated nematic phases that 
have constant number density in space but show modulation in orientational distribution. 
The prediction has been made very early \cite{Mayer,EuroPhys_56_247}. 
Later the twist-bend phase has been identified experimentally 
\cite{JMC_2002,pre_78_061705,prl_105_167801,prl_111_067801,Ntb,Ntb2}. 

The modulated nematic phases have also been discussed theoretically 
with different macroscopic phenomenological models \cite{prl82_940,pre2013,pre2014,pre2014_2,prl_3D_2014}. 
They are helpful to understanding the phase behaviors. 
However, since these models focus on particular phase transitions only, the order parameters and terms in the free energy are incomplete. 
In particular, all of these models do not include the biaxial nematic phase, which is studied separately in the literature (see \cite{biaxl2} and the references therein). 
Also, these models provide little information about the effect of molecular interaction on the phase transition. 
Some models incorporate microscopic interaction \cite{greco2014Molecular,tomczyk2016Twist}, but the desired phases are induced under artificial external forces that resemble the structure of the phases. 
On the other hand, molecular simulations 
\cite{JCP_111_9871,LC_29_483,PRE_67_011703,JCP_123_174907,JCP_129_244903} 
have also been carried out for bent-core molecules. 
A recent work \cite{PRL_115_147805} uses both molecular theory and molecular simulation to study a curved molecule that can exhibit the twist-bend phase. Molecular theory or molecular simulation can, indeed, build connection between molecular interaction and phase behaviors, but they are also costly in computation. 

Understanding the connection between the molecular interaction and the resulting phase behaviors is the fundamental problem for liquid crystals. 
It is more significant for bent-core molecules, because ample experimental results suggest that the phase behaviors of bent-core molecules can be sensitively dependent on specific molecular architecture \cite{JJAP}. 
To achieve this goal, it is necessary that we are armed with a model that 
(1) clearly reflects the role of molecular interaction; 
(2) can be solved efficiently, so that we can systematically examine the effect of physical parameters without spending very long time. 
Of all the models we mentioned above, microscopic models only meet the first requirement, while macroscopic models only meet the second, although some efforts are made to match both goals \cite{kusumaatmaja2015free,robinson2017molecular}.

In \cite{RodModel}, a tensor model is constructed for rod-like molecules, which takes a macroscopic form, while carrying information of the microscopic interaction. 
Starting from the molecular theory that includes the entropy and the pairwise interaction, the model is derived by the expansion of the spatial moments of the kernel function, along with the Bingham approximation \cite{bingham1974antipodally} that minimizes the entropy term with the value of second-order tensor fixed. 
By adopting the hard-core interaction and a simple molecular geometry, analytical calculations can be done in the expansion. 
In the resulting model, the free energy is expressed by some tensors, with the coefficients being functions of molecular parameters. 

For general cases where analytical calculations are not available, we have discussed the expansion for homogeneous phases \cite{SymmO}. 
First, we analyze the symmetries of the spatially homogeneous kernel function that originate from the molecular symmetry. 
Then, we are able to choose a finite dimensional polynomial space satisfying the symmetries. 
Since we can separate variables for each monomial, the free energy can be expressed by some tensors if we use any function in the polynomial space to approximate the kernel function. 
When the truncation criterion is fixed, the polynomial space is determined by the symmetries. 
Therefore, the form of free energy, as well as the tensors that appear in the free energy and serve as order parameters, is determined by molecular symmetry. 
For bent-core molecules, if we truncate at second order, the order parameters include three tensors, one first-order and two second-order. 
Finally, we calculate the projection of the kernel function in the polynomial space to derive the coefficients. In this way, the coefficients receive the information of molecular interaction in the kernel function, and are expressed as functions of molecular parameters. 

The purpose of this paper is to construct a tensor model for inhomogeneous phases. 
Now the kernel function is not spatially homogeneous, so we need to approximate its spatial moments as in \cite{RodModel}. 
In this case, the procedure for spatially homogeneous phases is still applicable with significant extensions technically. 
In particular, we need to find a suitable representation of spatial moments before writing down the approximation polynomial space. 
The resulting free energy is still a functional of the three tensors obtained for the homogeneous phases, but contains couplings and derivatives that enable us to study modulated nematic phases. 
For the entropy term, we follow the idea for rod-like molecules by minimizing it with the value of three tensors fixed. 
The model has the following features. 
\begin{itemize}
\item The form of free energy is determined by the molecular symmetry. 
Thus, the model is applicable to any molecule with the same symmetry. 
Moreover, the free energy is independent of the choice of the reference space-fixed orthonormal frame. 
\item Under certain truncation criterion, the model includes all the terms allowed by the molecular symmetry. 
Thus, the model is not specifically designed for certain phase transitions. 
\item For molecules with the same symmetry and different architecture or interaction, they are differentiated by the coefficients that are derived as functions of molecular parameters. 
\item As a special case, the model reduces to a model for rod-like molecules if the bending is straightened. 
\end{itemize}

We use the model to study the nematic phases of bent-core molecules resulting from the hard-core interaction, 
and find that the uniaxial and biaxial nematic phases, as well as the modulated twist-bend phase are all possible to occur, 
which cover all the nematic phases found experimentally so far. 
We obtain the phase diagram about molecular parameters, showing how the molecular parameters affect the modulation in the twist-bend phases. 
In addition, we examine the nematic phases of star molecules, a variant of bent-core molecules, to illustrate the effect of the molecular shape on the phase behavior. 
To our knowledge, it is the first result in which the phase behavior about the molecular shape is systematically examined in a theoretical model. 

The rest of paper is organized as follows. In Sec. \ref{model} we derive the tensor model from molecular theory. 
The numerical results are presented in Sec. \ref{results}. 
A concluding remark is given in Sec. \ref{concl}. 
Some details are given in Appendix. 

\section{The tensor model\label{model}}
\begin{figure}
\centering
\includegraphics[width=.3\textwidth,keepaspectratio]{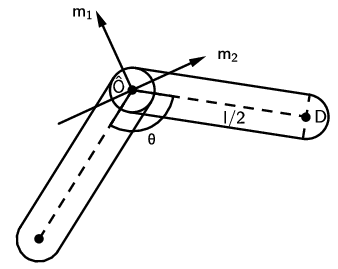}
\includegraphics[width=.3\textwidth,keepaspectratio]{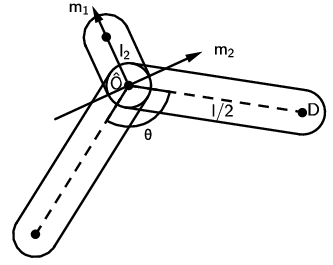}
\caption{\label{mol0}A bent-core molecule (left) and a star molecule (right). }
\end{figure}
\subsection{Notations}
We consider bent-core molecules and star molecules, drawn in Fig. \ref{mol0}. 
A bent-core molecule has two identical arm joint with fixed angle $\theta$. 
Each arm is a cylinder with two spherical caps, with the length $l/2$ and the diameter $D$. 
A star molecule has a third arm of the length $l_2$ along the arrowhead direction. 
Both molecules are regarded as fully rigid. 
Thus, the position and orientation of a molecule are represented by those of the 
orthonormal frame $(\hat{O};\bm{m}_1,\bm{m}_2,\bm{m}_3)$ mounted on it. 
As shown in Fig. \ref{mol0}, $\bm{m}_1$ points toward the arrowhead direction, and 
$\bm{m}_2$ is along the connection of the farther ends of two arms. 
Both molecules have the $C_{2v}$ symmetry, which allows the symmetry plane $\hat{O}\bm{m}_1\bm{m}_2$ and the twofold rotational symmetry round $\bm{m}_1$. 
Denote by $\bm{x}\in \mathbb{R}^3$ the position of $\hat{O}$ and by 
$P\in SO(3)$ the orientation of the frame. 
The matrix representation of $P$, which consists of the components of $\bm{m}_i$, 
can be expressed by Euler angles, 
\begin{align}
P=&(\bm{m}_1,\bm{m}_2,\bm{m}_3)
=\left(\begin{array}{ccc}
m_{11} & m_{21} & m_{31}\\
m_{12} & m_{22} & m_{32}\\
m_{13} & m_{23} & m_{33}
\end{array}\right)\nonumber\\
=&\left(
\begin{array}{ccc}
 \cos\alpha &\quad -\sin\alpha\cos\gamma &\quad\sin\alpha\sin\gamma\\
 \sin\alpha\cos\beta &\quad\cos\alpha\cos\beta\cos\gamma-\sin\beta\sin\gamma &
 \quad -\cos\alpha\cos\beta\sin\gamma-\sin\beta\cos\gamma\\
 \sin\alpha\sin\beta &\quad\cos\alpha\sin\beta\cos\gamma+\cos\beta\sin\gamma &
 \quad -\cos\alpha\sin\beta\sin\gamma+\cos\beta\cos\gamma
\end{array}
\right).\label{EulerRep}
\end{align}
The uniform probability measure on $SO(3)$ is given by
$$
\d P=\frac{1}{8\pi^2}\sin\alpha\d\alpha\d\beta\d\gamma. 
$$
We can also view $\bm{m}_i$ and $m_{ij}$ as functions of $P$. 
In what follows, we use the notation $\bm{m}_i(P)$ and $m_{ij}(P)$ 
to represent the $\bm{m}_i$ and $m_{ij}$ determined by a certain $P$. 

The summation over repeated indices will be used. 
The product $\bm{m}_1\bm{m}_1$ is recognized as tensor product and results in a second-order tensor, while $\bm{m}_1\cdot\bm{m}_2$ is the inner product. 
For a second-order tensor $Q$, we use $|Q|^2=Q:Q=Q_{ij}Q_{ij}$. 

\subsection{The derivation of tensor model}
Our starting point is the second virial expansion. 
The free energy includes the entropy and the contribution of pairwise molecular interaction, 
\begin{align}
\frac{F[f]}{\beta_0}=
\int \d P\d\bm{x} f(\bm{x},P)\log f(\bm{x},P)+\frac{1}{2}\int\d P \d\bm{x}\d P'\d\bm{x'}
f(\bm{x},P)G(\bm{r},P,P')f(\bm{x'},P'), \label{virial}
\end{align}
where $\bm{r}=\bm{x'}-\bm{x}$ is the relative position of two molecules. 
The energy is measured by $\beta_0$, the product of 
the Boltzmann constant and the temperature. 
The number density $f$ is a function of the position $\bm{x}$ and the orientation $P$. 
We define $c(\bm{x})=\int \d P f(\bm{x},P)$ as the spatial concentration, and 
${\rho}(\bm{x},P)=f(\bm{x},P)/c(\bm{x})$ as the orientational density. 
They satisfy 
$$
\int \d\bm{x}\d P f(\bm{x},P)=\int\d\bm{x}c(\bm{x})\int\d P\rho(\bm{x},P)=c_0V,
$$
where $V$ is the volume of the system, and $c_0$ is the average concentration. 
The kernel $G(\bm{r},P,P')$ is the Mayer function $G=1-\exp(-U/\beta_0)$ 
\cite{Mayerbook} about the pairwise potential $U$. 
In the case of hard-core potential, if two molecules touch, then 
$U(\bm{r},P,P')=+\infty$, leading to $G(\bm{r},P,P')=1$; 
otherwise $U(\bm{r},P,P')=0$, namely $G(\bm{r},P,P')=0$. 

To derive the form of the tensor model, 
we expand the pairwise interaction term in \eqref{virial} about $\bm{r}$ and $P$. 
After the expansion, we are able to express the pairwise interaction term by the three tensors identified in \cite{SymmO}. 
Then, we minimize the entropy term with the value of these tensors fixed, so that it is also expressed as a functional of the three tensors. 
This approach has also been adopted for rod-like molecules \cite{ball2010nematic,katriel1986free,RodModel}, where the density function becomes the Bingham distribution. 

\subsubsection{Spatial and orientational expansion}
First, we do Taylor expansion on $f(\bm{x'},P')=f(\bm{x}+\bm{r},P')$ with respect to $\bm{r}$, yielding 
\begin{align}
\frac{F[f]}{\beta_0}=&
  \int \d P\d\bm{x} f(\bm{x},P)\log f(\bm{x},P)\nonumber\\
  &+\sum_{k\ge 0}\frac{1}{2k!}\int\d\bm{x}\d P\d P'
    f(\bm{x},P)M^{(k)}(P,P')\nabla^k f(\bm{x},P'),\label{TlExp}
\end{align}
where 
\begin{equation}
  M^{(k)}(P,P')=\int G(\bm{r},P,P')
  \underbrace{\bm{r}\ldots\bm{r}}_{k\ \mbox{\small{times}}}\d\bm{r}, \label{spcmoment}
\end{equation}
a $k$th-order symmetric tensor, is the $k$th moment of $G$. 
For the hard-core interaction, the integration is taken on the region where $G=1$. 
By determining this region, we are able to calculate $M^{(k)}$ numerically. 
The detail is described in Appendix. 
Because the size of the region is proportional to $l^3$, 
we have $M^{(k)}\propto l^{k+3}$. 

Next, we expand $M^{(k)}(P,P')$ with respect to $P$ and $P'$. To clearly present the idea, we briefly review the expansion of $M^{(k)}$ for rod-like molecules discussed in \cite{RodModel}. 
In particular, we only look at $M^{(0)}$ and $M^{(2)}$ because they are sufficient for nematic phases (note that $M^{(1)}=0$). 
In this case, $M^{(k)}=M^{(k)}(\bm{m},\bm{m'})$ where $\bm{m}$ and $\bm{m'}$ are the directors of two rods (or, in the context of the current work, we may let $\bm{m}=\bm{m}_1(P)$ and $\bm{m'}=\bm{m}_1(P')$; see Theorem 3.4 in \cite{SymmO}). 
Analytical calculations give
\begin{align}
  M^{(0)}(\bm{m},\bm{m'})&=M^{(0)}(\eta), \label{rod0}\\
  M^{(2)}(\bm{m},\bm{m'})&=B_1(\eta)I+B_2(\eta)(\bm{m}\bm{m}+\bm{m'}\bm{m'})+B_3(\eta)(\bm{m}\bm{m'}+\bm{m'}\bm{m}), \label{rod2}
\end{align}
where $\eta=\bm{m}\cdot\bm{m'}$ is the inner product of the two directors, and $I$ is the identity matrix. 
Then, $M^{(0)}(\eta)$ and $B_i(\eta)$ are expanded as polynomials of $\eta$. 
In the resulting approximation formulas, $M^{(0)}$ and $M^{(2)}$ are expressed as polynomials of $\bm{m}$ and $\bm{m'}$. 
In this way, the variables $\bm{m}$ and $\bm{m'}$ are separated, leading to the approximate free energy as a function of tensors. 
An important point to be noted is the truncation of $M^{(0)}(\eta)$ and $B_i(\eta)$. 
The truncation is according to the order of each of $\bm{m}$ and $\bm{m'}$. 
Specifically, $M^{(0)}$ and $B_1$ are truncated at fourth order, $B_2$ at second order, and $B_3$ at third order. 
This is because $\bm{m}\bm{m}+\bm{m'}\bm{m'}$ contribute to the order by two, and $\bm{m}\bm{m'}+\bm{m'}\bm{m}$ contribute to the order by one. 
By this truncation, the approximation formulas include all the terms such that each of $\bm{m}$ and $\bm{m'}$ is not larger than fourth order, respectively. 
This truncation is adopted because the order of $\bm{m}$ and $\bm{m'}$ determines the order of tensor in the free energy. 
Under the above truncation, the corresponding free energy is a function including all the allowed terms of tensors up to fourth order. 

Returning to bent-core molecules, we aim to approximate each component of $M^{(k)}$ as a polynomial of $\bm{m}_{j}(P)$ and $\bm{m'}_{j}=\bm{m}_{j}(P')$. 
Similar to rod-like molecules \cite{RodModel}, we will only consider $k=0,1,2$, because we only examine nematic phases. 
Since we cannot do analytical calculations, we will follow the procedure in \cite{SymmO} with some extensions to determine the form of approximation formula by symmetric properties. 
The case $k=0$ has been discussed previously and will be reviewed shortly. 
For $k\ge 1$, we will first write down expressions similar to \eqref{rod2}, followed by polynomial approximations. 
For the truncation of the polynomials, we only retain the terms such that the degrees of $\bm{m}_{j}$ and $\bm{m}_{j}'$ are no more than second order, respectively. 
This makes the free energy as a function of tensors up to second order. 
The choice is largely based on keeping the model concise. 
One can also choose to truncate at fourth order like what is done for rod-like molecules, but at the expense of having over $100$ terms in the free energy. 

Denote the relative orientation and its components as 
\begin{equation}
\bar{P}=P^{-1}P'=(p_{ij})_{3\times 3}=(\bm{m}_i\cdot\bm{m'}_j)_{3\times 3}, \quad i,j=1,2,3, \label{relP}
\end{equation}
where we denote $\bm{m'}_{i}=\bm{m}_i(P')$. 
The following equalities shall be satisfied for molecules with the symmetry 
plane $\hat{O}\bm{m}_1\bm{m}_2$, 
\begin{align}
  G(T\bm{r},TP,TP')&=G(\bm{r},P,P'), \quad \forall T\in SO(3), \label{rotate0}\\
  G(-\bm{r},P',P)&=G(\bm{r},P,P'), \label{switch00}\\
  G(-\bm{r},PJ,P'J)&=G(\bm{r},P,P')\mbox{ for }J=\mbox{diag}(-1,-1,1). 
  \label{achiral}
\end{align}
The above equalities have been stated in \cite{SymmO}. 
The meaning of the three equalities is that $G$ is invariant 
when two molecules rotate together, when two molecules are switched, and when 
one molecule is reflected about the plane $\hat{O}\bm{m}_1\bm{m}_2$ of 
the other molecule. 

As the simplest case, we review the key points in the expansion of $M^{(0)}$. 
By setting $T=P^{-1}$ in (\ref{rotate0}), we can see that $M^{(0)}(P,P')=M^{(0)}(I,P^{-1}P')$. 
Thus, $M^{(0)}$ is a function of the relative orientation $\bar{P}$. 
Then from (\ref{achiral}), we deduce that 
\begin{equation}
M^{(0)}(\bar{P})=M^{(0)}(J\bar{P}J). \label{symplane0}
\end{equation}
Note that $\bar{P}$ and $J\bar{P}J$ are the only two elements in 
$SO(3)$ when $(p_{11},p_{12},p_{21},p_{22})$ is fixed. 
Hence $M^{(0)}$ is reduced to a function of the above four scalars. 
By (\ref{switch00}), we have $M^{(0)}(\bar{P})=M^{(0)}(\bar{P}^T)$, leading to 
\begin{equation}
  M^{(0)}(p_{11},p_{12},p_{21},p_{22})=M^{(0)}(p_{11},p_{21},p_{12},p_{22}). 
  \label{switch0}
\end{equation}
We use a polynomial of $p_{11},\ p_{12},\ p_{21},\ p_{22}$ to approximate $M^{(0)}$, denoted by 
$\hat{M}^{(0)}$. It shall satisfy \eqref{switch0} as well. 
Furthermore, it has the $\bm{m}_2\to -\bm{m}_2$ and 
$\bm{m'}_2\to -\bm{m'}_2$ symmetries. Since $p_{ij}=\bm{m}_i\cdot\bm{m'}_j$, 
only the terms where both $\bm{m}_2$ and $\bm{m'}_2$ appear even times 
can be retained. For example, the term $p_{22}=\bm{m}_2\cdot\bm{m'}_2$ will be 
discarded since both $\bm{m}_2$ and $\bm{m'}_2$ appear one time. 
Thus, we obtain the following quadratic approximation 
\begin{equation}\label{quadapp}
\hat{M}^{(0)}=c_{00}+c_{01}p_{11}+c_{02}p_{11}^2+c_{03}p_{22}^2+c_{04}(p_{12}^2+p_{21}^2). 
\end{equation}
The first index of the coefficients $c_{0j}$ is zero, 
corresponding to the zeroth moment $M^{(0)}$. 
These coefficients are independent of $P$ and $P'$. 

When we apply the above procedure to the expansion of $M^{(k)}$ for $k\ge 1$, 
some modifications need to be made since $M^{(k)}$ is a $k$th-order tensor. 
We will first seek a representation similar to \eqref{rod2}. 
The representation \eqref{rod2} conveys two messages: 
\begin{enumerate}
\item $M^{(2)}$ can be expressed as linear combination of some tensors generated by $\bm{m}$ and $\bm{m'}$ and $I$. Moreover, the expression is symmetric about $\bm{m}$ and $\bm{m'}$. 
\item The coefficients $B_i$ depend only on $\bm{m}\cdot\bm{m'}$ that describes the relative orientation of $\bm{m}$ and $\bm{m'}$. 
\end{enumerate}
Actually, these two statements hold for any $M^{(k)}$ (see \cite{RodModel} (3.22) for the fourth moment $M^{(4)}$). 
With this observation, for bent-core molecules, we first seek for a representation of $M^{(k)}$ by linear combination of tensors generated by $\bm{m}_i(P)$, $\bm{m'}_i=\bm{m}_i(P')$ and $I$, which is symmetric about $P$ and $P'$. 
Then, we figure out the symmetry of the coefficients in this representation, followed by polynomial approximation. 
We pay particular attention to the effect of the symmetry plane $\hat{O}\bm{m}_1\bm{m}_2$. It eliminates the appearance of $\bm{m}_3$ and $\bm{m'}_3$ in $M^{(0)}$ by \eqref{symplane0}. We can also eliminate them in any $M^{(k)}$, as we will show below. 
We will only discuss $M^{(1)}$ in detail, because $M^{(2)}$ follows the same way. 

Now we start to discuss $M^{(1)}$. 

\textbf{Step 1.} We show that for fixed $(P,P')$, $M^{(1)}$ can be expressed as 
\begin{equation}
  M^{(1)}(P,P')=\tilde{c}_1\bm{m}_1+\tilde{c}_2\bm{m}_2
  +\tilde{c}_{1'}\bm{m'}_1+\tilde{c}_{2'}\bm{m'}_2, 
  \label{firstM_1}
\end{equation}
where $\tilde{c}_j\,(j=1,2,1',2')$ are functions of $(P,P')$. 

We begin with writing 
\begin{equation}
  M^{(1)}(P,P')=\tilde{c}_1\bm{m}_1+\tilde{c}_2\bm{m}_2+\tilde{c}_3\bm{m}_3. 
  \label{firstM_0}
\end{equation}
To write down a representation symmetric about $P$ and $P'$, 
we express $\bm{m}_3$ by linear combination of $\bm{m}_1$, $\bm{m}_2$, $\bm{m'}_1$ and $\bm{m'}_2$. 
Note that this cannot be done when $\bm{m}_3=\pm \bm{m'}_3$, because in this case $\text{span}\{\bm{m}_1,\bm{m}_2,\bm{m'}_1,\bm{m'}_2\}=\text{span}\{\bm{m}_1,\bm{m}_2\}$, and $\bm{m}_3\notin \text{span}\{\bm{m}_1,\bm{m}_2\}$. 
However, if the molecule has the symmetry plane $\hat{O}\bm{m}_1\bm{m}_2$, we show that $\tilde{c_3}=0$ when $\bm{m}_3=\pm \bm{m'}_3$. 
Actually, the condition $\bm{m}_3=\pm \bm{m'}_3$ can be rewritten as $\bar{P}=P^{-1}P'=\mbox{diag}(W,\pm 1)$ where $W$ is a $2\times 2$ orthogonal matrix. 
Thus, we have $J\bar{P}J=\bar{P}$ with $J=\mbox{diag}(-1,-1,1)$. 
Together with \eqref{rotate0} and (\ref{achiral}), we have 
\begin{align}
  G(\bm{r},P,P')=&G(P^{-1}\bm{r},I,\bar{P})=G(-P^{-1}\bm{r},J,\bar{P}J)\nonumber\\
  =&G(-JP^{-1}\bm{r},I,J\bar{P}J)=G(-JP^{-1}\bm{r},I,\bar{P})=G(-PJP^{-1}\bm{r},P,P'). \label{symplane}
\end{align}
Note that $-PJP^{-1}\bm{r}=\bm{r}-2(\bm{r}\cdot\bm{m}_3)\bm{m}_3$. 
Taking \eqref{symplane} into \eqref{spcmoment}, we obtain 
\begin{align}
  2\tilde{c}_3=&2\bm{m}_3\cdot M^{(1)}(P,P')=2\int G(\bm{r},P,P')(\bm{r}\cdot\bm{m}_3)\d\bm{r}\nonumber\\
  =&\int \big(G(\bm{r},P,P')+G(-PJP^{-1}\bm{r},P,P')\big)(\bm{r}\cdot\bm{m}_3)\d\bm{r}\nonumber\\
  =&\int G(\bm{r},P,P')(\bm{r}-PJP^{-1}\bm{r})\cdot\bm{m}_3\d\bm{r}\nonumber\\
  =&0. \label{symplane1}
\end{align}
Therefore, when the molecule has the symmetry plane $\hat{O}\bm{m}_1\bm{m}_2$, we are allowed to use the representation \eqref{firstM_1}. 

\textit{Remark.} Whether the molecular has a symmetry plane affects the form of the representation of $M^{(k)}$. 
If the molecule is chiral, we have to include $\bm{m}_3$. It is also the case for rod-like molecules. 
Actually, we have $M^{(1)}\ne 0$ even if a chiral rod-like molecule has the head-to-tail symmetry. 

\textbf{Step 2.}
We analyze the symmetric properties of the scalars $\tilde{c}_j$ in \eqref{firstM_1}. 
Apparently, the representation \eqref{firstM_1} is not unique. 
In what follows, when we say that $\tilde{c}_j$ satisfy certain symmetries, it means that there exists a representation in which the symmetries hold. 

First, we can require that they are functions of $\bar{P}$. 
We deduce from (\ref{rotate0}) that for any $T\in SO(3)$, 
\begin{align}
M^{(1)}(TP,TP')=&\int \bm{r}G(\bm{r},TP,TP')\d\bm{r}=\int \bm{r}G(T^{-1}\bm{r},P,P')\d\bm{r}\nonumber\\
=&\int (T\bm{r})G(\bm{r},P,P')\d\bm{r}=TM^{(1)}(P,P'). \nonumber
\end{align}
It implies that if we already know the value of $\tilde{c}_j(I,P^{-1}P)$, we may let $\tilde{c}_j(P,P')=\tilde{c}_j(I,P^{-1}P')$ in \eqref{firstM_1} to obtain a representation. From now on, we will omit this kind of explanations and just write 
\begin{equation}
\tilde{c}_{j}(TP,TP')=\tilde{c}_{j}(P,P')=\tilde{c}_{j}(I,P^{-1}P')\triangleq\tilde{c}_{j}(\bar{P}),\ j=1,2,1',2'. \label{rotate1}
\end{equation}
Next, we substitute $(P,P')$ with $(PJ,P'J)$ in \eqref{firstM_1}. 
Using (\ref{achiral}), we obtain 
\begin{align*}
  &M^{(1)}(PJ,P'J)=-M^{(1)}(P,P')\\
  =&\tilde{c}_1(J\bar{P}J)(-\bm{m}_1)+\tilde{c}_2(J\bar{P}J)(-\bm{m}_2)
  +\tilde{c}_{1'}(J\bar{P}J)(-\bm{m'}_1)+\tilde{c}_{2'}(J\bar{P}J)(-\bm{m'}_2) \\
  =&-\tilde{c}_1(\bar{P})\bm{m}_1-\tilde{c}_2(\bar{P})\bm{m}_2
  -\tilde{c}_{1'}(\bar{P})\bm{m'}_1-\tilde{c}_{2'}(\bar{P})\bm{m'}_2, 
\end{align*}
yielding (cf. \eqref{symplane0})
\begin{equation}
\tilde{c}_{j}(\bar{P})=\tilde{c}_{j}(J\bar{P}J)\triangleq\tilde{c}_{j}(p_{11},p_{12},p_{21},p_{22}),\ j=1,2,1',2'. \label{achiral1}
\end{equation}
At this point, we have eliminated the appearance of $\bm{m}_3$ and $\bm{m'}_3$ in \eqref{firstM_1}. 
Then, we switch $P$ and $P'$ in \eqref{firstM_1}. By (\ref{switch00}), we have 
\begin{align*}
  &M^{(1)}(P',P)=-M^{(1)}(P,P')\\
  =&\tilde{c}_1(\bar{P}^T)\bm{m'}_1+\tilde{c}_2(\bar{P}^T)\bm{m'}_2
  +\tilde{c}_{1'}(\bar{P}^T)\bm{m}_1+\tilde{c}_{2'}(\bar{P}^T)\bm{m}_2 \\
  =&-\tilde{c}_1(\bar{P})\bm{m}_1-\tilde{c}_2(\bar{P})\bm{m}_2
  -\tilde{c}_{1'}(\bar{P})\bm{m'}_1-\tilde{c}_{2'}(\bar{P})\bm{m'}_2. 
\end{align*}
Thus, we have $\tilde{c}_1(\bar{P})=-\tilde{c}_{1'}(\bar{P}^T)$ and $\tilde{c}_2(\bar{P})=-\tilde{c}_{2'}(\bar{P}^T)$, leading to (cf. \eqref{switch0})
\begin{align}
  &\tilde{c}_{1}(p_{11},p_{12},p_{21},p_{22})=-\tilde{c}_{1'}(p_{11},p_{21},p_{12},p_{22}), \nonumber\\
  &\tilde{c}_{2}(p_{11},p_{12},p_{21},p_{22})=-\tilde{c}_{2'}(p_{11},p_{21},p_{12},p_{22}). \label{switch1}
\end{align}

\textbf{Step 3.}
With the symmetric properties \eqref{switch1}, we write down the polynomial approximation of $\tilde{c}_{j}$ with attention to the 
$\bm{m}_2\to -\bm{m}_2$ and $\bm{m'}_2\to -\bm{m'}_2$ symmetries. 
The degree of polynomial is chosen such that both $\bm{m}_i$ and $\bm{m'}_i$ 
are truncated at second order in \eqref{firstM_1}. 
For example, the term $p_{21}\bm{m}_2$ can be rewritten as 
$(\bm{m}_2\cdot\bm{m}_1')\bm{m}_2$, in which the order of $\bm{m}_2$ is two. 
As an example, we look into $\tilde{c}_{2}$. 
The $\bm{m}_2\to -\bm{m}_2$ symmetry allows only one term $ap_{21}$ in the polynomial approximation, where $a$ is the coefficient. 
Similarly, in the polynomial approximation of $\tilde{c}_{2'}$, there is also only one term $a'p_{12}$. 
Then we use \eqref{switch1} to arrive at $a=-a'$. 
In this way, the polynomial approximations are written as 
\begin{align}
  \tilde{c}_1=-c_{10}-c_{11}p_{11}, \quad
  \tilde{c}_{1'}=c_{10}+c_{11}p_{11}, \quad
  \tilde{c}_{2}=-c_{12}p_{21}, \quad
  \tilde{c}_{2'}=c_{12}p_{12}, \label{poly1}
\end{align}
where we denote the coefficients by $c_{1j}$. 
The first index of $c_{1j}$ becomes one to indicate that they come from $M^{(1)}$. 
The coefficients $c_{1j}$ are independent of $P$ and $P'$. 

The expansion of $M^{(2)}$ follows the same way as $M^{(1)}$ and is described briefly. 
We start from 
\begin{align}
M^{(2)}(P,P')=\sum_{l_1,l_2=1,2,3}\tilde{c}_{l_1l_2}\bm{m}_{l_1}\bm{m}_{l_2}. 
\nonumber
\end{align}
Then, we express $\bm{m}_3$ by linear combination of $\bm{m}_1$, $\bm{m}_2$, $\bm{m'}_1$ and $\bm{m'}_2$ if $\bm{m'}_3\ne\pm\bm{m}_3$. In the case $\bm{m'}_3=\pm\bm{m}_3$, we use \eqref{symplane} to obtain $\tilde{c}_{13}=\tilde{c}_{23}=\tilde{c}_{31}=\tilde{c}_{32}=0$ (cf. \eqref{symplane1}), and utilize the equality $\bm{m}_3\bm{m}_3=I-\bm{m}_1\bm{m}_1-\bm{m}_2\bm{m}_2$ to take care of the term $\bm{m}_3\bm{m}_3$. 
In any of the above two cases, we are allowed to use the following representation that is symmetric about $P$ and $P'$, 
\begin{align}
M^{(2)}(P,P')=&\tilde{c}_{00'}I
+\sum_{l_1,l_2=1,2}\tilde{c}_{l_1l_2}\bm{m}_{l_1}\bm{m}_{l_2}
+\sum_{l'_1,l'_2=1',2'}\tilde{c}_{l'_1l'_2}\bm{m'}_{l_1}\bm{m'}_{l_2}\nonumber\\
&+\sum_{l=1,2,l'=1',2'}\tilde{c}_{ll'}(\bm{m}_{l}\bm{m'}_{l'}+\bm{m'}_{l'}\bm{m}_{l}), 
\label{sep2}
\end{align}
where we require
\begin{equation}
\tilde{c}_{l_1l_2}=\tilde{c}_{l_2l_1},~\tilde{c}_{l'_1l'_2}=\tilde{c}_{l'_2l'_1} \nonumber
\end{equation}
because $M^{(2)}$ is symmetric. 
Repeating the derivation of \eqref{rotate1} and \eqref{achiral1} for $M^{(1)}$, we can deduce that $\tilde{c}_{j_1j_2}$ are functions of $(p_{11},p_{12},p_{21},p_{22})$. 
Then, by switching $P$ and $P'$ in \eqref{sep2} and using (\ref{switch00}), we obtain (cf. \eqref{switch1})
\begin{align}
  \tilde{c}_{00'}(p_{11},p_{12},p_{21},p_{22})
  &=\tilde{c}_{00'}(p_{11},p_{21},p_{12},p_{22}), \nonumber\\
  \tilde{c}_{11}(p_{11},p_{12},p_{21},p_{22})
  &=\tilde{c}_{1'1'}(p_{11},p_{21},p_{12},p_{22}), \nonumber\\
  \tilde{c}_{12}(p_{11},p_{12},p_{21},p_{22})
  &=\tilde{c}_{1'2'}(p_{11},p_{21},p_{12},p_{22}), \nonumber\\
  \tilde{c}_{22}(p_{11},p_{12},p_{21},p_{22})
  &=\tilde{c}_{2'2'}(p_{11},p_{21},p_{12},p_{22}), \nonumber\\
  \tilde{c}_{12'}(p_{11},p_{12},p_{21},p_{22})
  &=\tilde{c}_{21'}(p_{11},p_{21},p_{12},p_{22}). \label{switch2}
\end{align}
By noting the $\bm{m}_2\to -\bm{m}_2$ and $\bm{m'}_2\to -\bm{m'}_2$ symmetries, 
and keeping the truncation at second order for both $\bm{m}_i$ and $\bm{m'}_i$ in \eqref{sep2}, 
we obtain the polynomial approximations of $\tilde{c}_{j_1j_2}$, 
\begin{align}
  \tilde{c}_{00'}&=-c_{20}-c_{21}p_{11}-c_{22}p_{11}^2-c_{23}p_{22}^2
    -c_{24}(p_{12}^2+p_{21}^2), \nonumber\\
  \tilde{c}_{11}&=\tilde{c}_{1'1'}=-c_{25}, \nonumber\\
  \tilde{c}_{22}&=\tilde{c}_{2'2'}=-c_{26}, \nonumber\\
  \tilde{c}_{11'}&=-c_{27}-c_{28}p_{11}, \nonumber\\
  \tilde{c}_{22'}&=-c_{29}p_{22}, \nonumber\\
  \tilde{c}_{12}&=\tilde{c}_{21}=\tilde{c}_{1'2'}=\tilde{c}_{2'1'}=0, \nonumber\\
  \tilde{c}_{12'}&=-c_{2,10}p_{12}, \quad
  \tilde{c}_{21'}=-c_{2,10}p_{21}. \label{poly2}
\end{align}
Just as the notation for $M^{(0)}$ and $M^{(1)}$, the first index of $c_{2j}$ is two. 
Again the coefficients $c_{2j}$ do not depend on $P$ and $P'$. 

Summarizing \eqref{quadapp}, \eqref{firstM_1}, \eqref{poly1}, \eqref{sep2}, and \eqref{poly2}, we obtain the expansion of $M^{(0)}$, $M^{(1)}$, $M^{(2)}$, denoted by $\hat{M}^{(0)}$, $\hat{M}^{(1)}$, $\hat{M}^{(2)}$, 
\begin{align}
  \hat{M}&^{(0)}=c_{00}+c_{01}p_{11}+c_{02}p_{11}^2+c_{03}p_{22}^2+c_{04}(p_{12}^2+p_{21}^2),
  \nonumber\\
  \hat{M}&^{(1)}=(-c_{10}-c_{11}p_{11})(\bm{m}_1-\bm{m'}_1)
  -c_{12}(p_{21}\bm{m}_2-p_{12}\bm{m'}_2),\nonumber\\
  \hat{M}&^{(2)}=-\big(c_{20}+c_{21}p_{11}+c_{22}p_{11}^2
  +c_{23}p_{22}^2+c_{24}(p_{12}^2+p_{21}^2)\big)I\nonumber\\
  &-c_{25}(\bm{m}_1\bm{m}_1+\bm{m'}_1\bm{m'}_1)
  -c_{26}(\bm{m}_2\bm{m}_2+\bm{m'}_2\bm{m'}_2)\nonumber\\
  &-(c_{27}+c_{28}p_{11})(\bm{m}_1\bm{m'}_1+\bm{m'}_1\bm{m}_1)\nonumber\\
  &-c_{29}p_{22}(\bm{m}_2\bm{m'}_2+\bm{m'}_2\bm{m}_2)\nonumber\\
  &-c_{2,10}\left[p_{12}(\bm{m}_1\bm{m'}_2+\bm{m'}_2\bm{m}_1)
  +p_{21}(\bm{m}_2\bm{m'}_1+\bm{m'}_1\bm{m}_2)\right]. \label{moments}
\end{align}

We substitute $M^{(k)}$ with $\hat{M}^{(k)}$ in (\ref{TlExp}). 
The purpose is to separate the the variables $P$ and $P'$. 
In this way, each term in $\hat{M}^{(k)}$ corresponds to a term in the free energy. Moreover, each term can be expressed by three tensors $\bm{p}$, $Q_1$, $Q_2$, define as 
\begin{equation}
\bm{p}=\left<\bm{m}_1\right>,\ Q_1=\left<\bm{m}_1\bm{m}_1\right>,\ 
Q_2=\left<\bm{m}_2\bm{m}_2\right>, \label{tensors}
\end{equation}
where $\left<\cdot\right>=\int\d P(\cdot)\rho(P)$ denotes the average about the orientational density $\rho$. 
As an example, the term $-p^2_{12}I$ in $\hat{M}^{(2)}$ generates the term $\nabla(cQ_1):\nabla(cQ_2)$, 
\begin{align}
  &\int\d\bm{x}\d P\d P'\, -p^2_{12}I: f(\bm{x},P)
  \nabla^2 f(\bm{x},P')\nonumber\\
  =&-\int\d\bm{x}\left(c(\bm{x})\int\d Pm_{1i}m_{1j}{\rho}(\bm{x},P)\right)
  \partial_{kk}\left(c(\bm{x})\int\d P'm'_{2i}m'_{2j}{\rho}(\bm{x},P')\right)\nonumber\\
  =&-\int\d\bm{x}\left(c(\bm{x})\left<m_{1i}m_{1j}\right>\right)
  \partial_{kk}\left(c(\bm{x})\left<m_{2i}m_{2j}\right>\right), \nonumber\\
  =&\int\d\bm{x}\partial_k\left(c(\bm{x})Q_{1ij}\right)
  \partial_{k}\left(c(\bm{x})Q_{2ij}\right). \label{Septerm}
\end{align}
Here we have done integration by parts and assume that the boundary terms vanish. 
We will write down all the terms afterwards in \eqref{FreeEng}. 
One shall also observe that the three tensors are all the nontrivial tensors about $\bm{m}_1$ and $\bm{m}_2$ up to second order. 

Finally, we point out that the derivation described in this section is applicable to any $M^{(k)}$, which is necessary if we aim to model smectic and columnar phases. 

\subsubsection{The entropy term}
By the expansion discussed above, we have defined three tensors as order parameters. 
Now we can express the entropy term as a function of these tensors by a constrained minimization problem (cf. \cite{ball2010nematic}). 
The entropy term can be rewritten as 
$$
\int\d\bm{x}\d P\,c\rho(\log c+\log\rho)=
\int\d\bm{x}c\log c+\int\d\bm{x}\left(c(\bm{x})\int\d P \rho\log\rho\right). 
$$
We minimize $\int\d P \rho\log \rho$ with the values of $\bm{p}$, $Q_1$ and $Q_2$ fixed. The Euler-Lagrange equation is written as 
\begin{equation}
  1+\log\rho=\lambda+\bm{b}\cdot\bm{m}_1+B_1:\bm{m}_1\bm{m}_1+B_2:\bm{m}_2\bm{m}_2, 
\end{equation}
where the Lagrange multipliers are chosen such that 
\begin{align}
    \int \d P \rho(P)=1,\  \int \d P\bm{m}_1\rho(P)=\bm{p},\ 
    \int \d P\bm{m}_1\bm{m}_1\rho(P)=Q_1,\
    \int \d P\bm{m}_2\bm{m}_2 \rho(P)=Q_2.
\end{align}
The Euler-Lagrange equation gives the Boltzmann distribution, 
\begin{equation}
\rho=\frac{1}{Z}\exp(\bm{b}\cdot\bm{m}_1+B_1:\bm{m}_1\bm{m}_1+B_2:\bm{m}_2\bm{m}_2),
\label{distribution}
\end{equation}
where $Z$ is the normalization factor, 
\begin{equation}
Z=\int\d P\exp(\bm{b}\cdot\bm{m}_1+B_1:\bm{m}_1\bm{m}_1+B_2:\bm{m}_2\bm{m}_2).
\label{PartFun}
\end{equation}

We require that $Q_1$ and $Q_2$ share an eigenframe and that $\bm{p}$ is their eigenvector. 
This approximation comes from a theoretical result for homogeneous phases \cite{XuCMS}. 
In other words, we assume that there exists a $T=(\bm{n}_1,\bm{n}_2,\bm{n}_3)\in SO(3)$ such that 
\begin{align}
\bm{p}=s\bm{n}_1, \
Q_1=q_{11}\bm{n}_1\bm{n}_1+q_{12}\bm{n}_2\bm{n}_2+q_{13}\bm{n}_3\bm{n}_3,\ 
Q_2=q_{21}\bm{n}_1\bm{n}_1+q_{22}\bm{n}_2\bm{n}_2+q_{23}\bm{n}_3\bm{n}_3, \label{codiag}
\end{align}
with $q_{i3}=1-q_{i1}-q_{i2}$. The eigenvalues shall satisfy 
\begin{align}
  &q_{ij}>0, s^2<q_{11}, \nonumber\\
  &q_{11}+q_{12},q_{11}+q_{21},q_{12}+q_{22},q_{21}+q_{22}<1,\nonumber\\
  &q_{11}+q_{12}+q_{21}+q_{22}>1. \label{range}
\end{align}
They originate from $Q_1-\bm{p}\bm{p}$, $Q_2$, $I-Q_1-Q_2$ are positive definite.
They originate from $Q_1-\bm{p}\bm{p}$, $Q_2$, $I-Q_1-Q_2$ are positive definite.
Furthermore, if $(s,q_{ij})$ lies in a subregion given by the above constraints, 
there exists a unique $(\bm{b},B_1,B_2)$ of the form 
\begin{align}
\bm{b}&=T(b_1,0,0)^T, 
B_1=T\mbox{diag}(b_{11},b_{12},0)T^T,
B_2=T\mbox{diag}(b_{21},b_{22},0)T^T,\label{BRep}
\end{align}
such that the moments of the corresponding Boltzmann distribution are $(\bm{p},Q_1,Q_2)$. 
We state and prove the result rigorously in Appendix (cf. Theorem \ref{clos_exist}). 

\subsubsection{The free energy}
Because we focus on nematic phases, we assume that $c(\bm{x})=c_0$ is constant, still denoted by $c$. 
By \eqref{TlExp}, \eqref{moments}, \eqref{distribution} and integration by parts, the tensor model is written as follows, 
\begin{align}
  \frac{F[\bm{p},Q_1,Q_2]}{\beta_0}=&\int\d\bm{x} 
  \Big\{c(\bm{b}\cdot\bm{p}+B_1:Q_1+B_2:Q_2-\log Z)\nonumber\\
  &+\frac{c^2}{2}(c_{01}|\bm{p}|^2+c_{02}|Q_1|^2+c_{03}|Q_2|^2
  +2c_{04}Q_1:Q_2)\nonumber\\
  &+c^2(c_{11}p_j\partial_iQ_{1ij}
  +c_{12}p_j\partial_iQ_{2ij})\nonumber\\
  &+\frac{c^2}{4}\left[c_{21}|\nabla\bm{p}|^2
  +c_{22} |\nabla Q_1|^2 + c_{23}|\nabla Q_2|^2
  + 2c_{24}\partial_{i}Q_{1jk}\partial_iQ_{2jk}\right.\nonumber\\
  &+ 2c_{27}\partial_{i}p_i\partial_jp_j
  + 2c_{28}\partial_iQ_{1ik}\partial_jQ_{1jk}\nonumber\\
  &+ \left.2c_{29}\partial_iQ_{2ik}\partial_jQ_{2jk}
  + 4c_{2,10}\partial_iQ_{1ik}\partial_jQ_{2jk}\right]\Big\}, \label{FreeEng}
\end{align}
where the components of $\bm{p}$ and $Q_k$ are denoted as $p_i$ and $Q_{kij}$. 
The first line comes from the entropy term. The second line comes from $\hat{M}^{(0)}$. The third line comes from $\hat{M}^{(1)}$, referred to as first-order elastic energy. They are crucial for modulated nematic phases to emerge. 
The rest terms come from $\hat{M}^{(2)}$, referred to as second-order elastic energy. 

\subsection{The coefficients}
Now we describe how to calculate the coefficients in \eqref{FreeEng}. 
We emphasize that the coefficients in the free energy are just those in $\hat{M}^{(k)}$. 
Note that $\hat{M}^{(k)}$ is the approximation of $M^{(k)}$ that is determined by molecular parameters. 
Hence we minimize the distance between $\hat{M}^{(k)}$ and $M^{(k)}$, defined as 
\begin{align}
\int_{SO(3)}\d P \d P'||M^{(k)}(P,P';l,D,\theta)-\hat{M}^{(k)}(P,P';\{c_{kj}\})||_F^2,\label{coe_ls}
\end{align}
where $||\cdot||_F$ is the Frobenius norm
$||M||_F^2=\sum_{i_1\ldots i_k} |M_{i_1\ldots i_k}|^2$. 
By solving this linear least-square problem, we can express $c_{kj}$ as 
functions of the molecular parameters $l$, $D$ and $\theta$. 
Furthermore, we have $c_{kj}\propto l^{k+3}$ because $M^{(k)}$ has the same scaling. 
Therefore, we can further nondimensionalize the model by the substitution 
$
\bar{\bm{x}}={\bm{x}}/{l},\ \bar{c}=cl^3,\ \bar{c}_{kj}={c_{kj}}/{l^{k+3}}.
$
Now $\bar{c}_{kj}$ become functions of two dimensionless parameters 
$\eta=D/l$ and $\theta$. 
For star molecules, $\hat{M}^{(k)}$ also depends on $l_2$, 
thus $\bar{c}_{kj}$ are also functions of $l_2/l$. 
For convenience, we still express these dimensionless quantities by the original notations. 

The second-order elastic energy shall be positive definite to ensure the 
lower-boundedness of the free energy. This can be guaranteed if the following inequalities hold, 
\begin{align}
  &c_{21},c_{22},c_{23},2c_{27}+c_{21},2c_{28}+c_{22},2c_{29}+c_{23}\ge 0,\nonumber\\ 
  &c_{24}^2\le c_{22}c_{23}, \nonumber\\
  &(2c_{2,10}+c_{24})^2\le (2c_{28}+c_{22})(2c_{29}+c_{23}). \label{coercive}
\end{align}
These inequalities can be easily observed after we rewrite the $c_{2j}$ terms for $7\le j\le 10$, for which we explain by the term $c_{27}$. 
First, we observe that $\partial_ip_i\partial_jp_j-\partial_jp_i\partial_ip_j$ is a boundary term. 
Thus, we can substitute $\partial_ip_i\partial_jp_j$ with $\partial_jp_i\partial_ip_j$ by doing integration by parts, and assume that the boundary term vanishes 
by adpoting suitable boundary conditions (for example periodic boundary conditions that will be used later). 
Then, we write 
$$
\partial_ip_{j}\partial_ip_{j}=\partial_ip_{j}\partial_jp_{i}+
\frac{1}{2}|\partial_ip_{j}-\partial_jp_{i}|^2. 
$$
By doing the same thing to the other three terms, the second-order elastic energy becomes 
\begin{align*}
  &\int\d\bm{x}\,\frac{c^2}{4}\big[\frac{1}{2}c_{21}|\partial_ip_j-\partial_jp_i|^2
    +(2c_{27}+c_{21})(\partial_ip_i)^2
  +\frac{1}{2}c_{22}|\partial_iQ_{1jk}-\partial_jQ_{1ik}|^2\\
  &+\frac{1}{2}c_{23}|\partial_iQ_{2jk}-\partial_jQ_{2ik}|^2
  +c_{24}(\partial_iQ_{1jk}-\partial_jQ_{1ik})(\partial_iQ_{2jk}-\partial_j{Q_{2ik}})\\
  &+ (2c_{28}+c_{22})|\partial_iQ_{1ik}|^2
  + (2c_{29}+c_{23})|\partial_iQ_{2ik}|^2
  + 2(2c_{2,10}+c_{24})\partial_iQ_{1ik}\partial_j(Q_{2jk})\big]. 
\end{align*}
Moreover, if $(2c_{2,10}+c_{24})^2< (2c_{28}+c_{22})(2c_{29}+c_{23})$ or $c_{21}>0$, 
it controls the first-order elastic energy. For example, we have 
\begin{align*}
  (2c_{28}+c_{22})|\partial_iQ_{1ik}|^2
  + (2c_{29}+c_{23})|\partial_iQ_{2ik}|^2
  + &2(2c_{2,10}+c_{24})\partial_iQ_{1ik}\partial_j(Q_{2jk})\\
  &-4\partial_ip_j(c_{11}Q_{1ij}+c_{12}Q_{2ij})
  \ge -C|\bm{p}|^2 
\end{align*}
for $C$ large enough, 
and the right-hand side is bounded from below since $|\bm{p}|<1$. 

When $\theta=\pi$, the molecule becomes a rod. 
In this case, all the coefficients involving $\bm{p}$ and $Q_1$ shall be zero, which is verified in our numerical calculation. 
Furthermore, it can be shown with Theorem \ref{clos_exist} that in the entropy term we have $\bm{b}=B_1=0$, for which we omit the detail. 
Thus, the free energy depends only on $Q_2$, written as 
\begin{align}
  \frac{F[Q_2]}{\beta_0}=&\int\d\bm{x} 
  \Big\{c(B_2:Q_2-\log Z)+\frac{c^2}{2}c_{03}|Q_2|^2
  +\frac{c^2}{4}\left(c_{23}|\nabla Q_2|^2
  + 2c_{29}\partial_iQ_{2ik}\partial_jQ_{2jk}\right)\Big\}. \label{FreeEng_rod}
\end{align}
It is a simplified version of the model proposed for rod-like molecules in \cite{RodModel} (see (3.15) in \cite{RodModel}), including only the terms involving the second-order tensor. 
The condition \eqref{coercive} becomes $c_{23},2c_{29}+c_{23}\ge 0$. 
It is weaker than what is proposed in \cite{longa1987extension} for rod-like molecules. 
This is because in \cite{longa1987extension}, the derivatives of the tensors are viewed as independent functions, and the coercivity is assumed pointwise, which is a stronger condition than requiring the free energy to be lower bounded. 

\begin{figure}
\centering
\includegraphics[width=.48\textwidth,keepaspectratio]{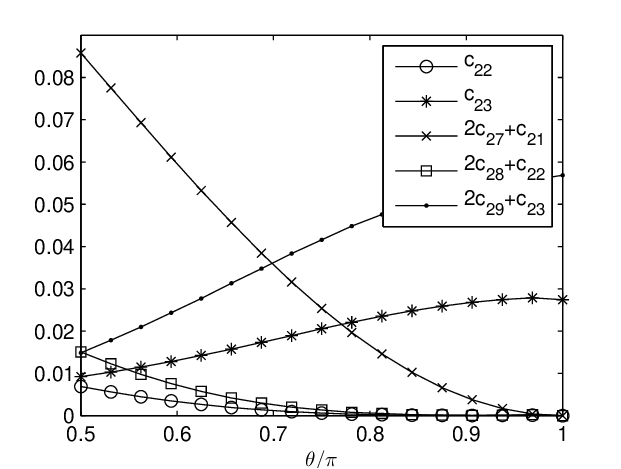}
\includegraphics[width=.48\textwidth,keepaspectratio]{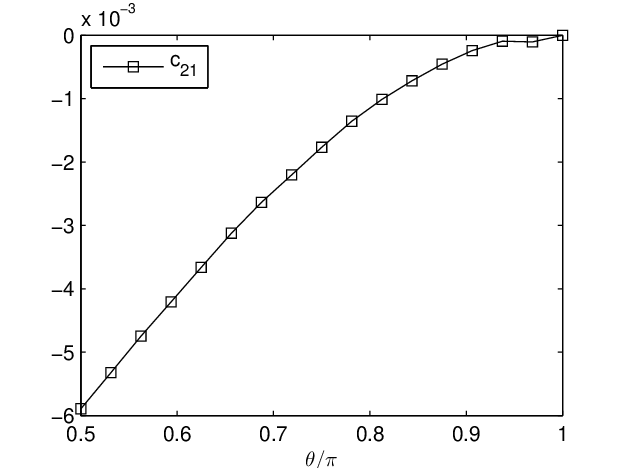}
\caption{\label{coef}The coefficients $c_{2j}$ in the second-order elastic energy for bent-core molecules, measured in the unit $(l/2)^5$, as functions of the bending angle $\theta$ when $\eta=1/40$. }
\end{figure}
Now we examine the coefficients calculated from \eqref{coe_ls}. 
The coefficients of the square terms for bent-core molecules are plotted in Fig. \ref{coef}. 
When $\theta=\pi$, we have $c_{23},2c_{29}+c_{23}>0$, and all the other $c_{2j}$ are zero. 
As the bending angle decreases, $c_{23}$ and $2c_{29}+c_{23}$ monotonely decrease, and the absolute values of other $c_{2j}$ increase monotonely, thus do not change sign. 
We find that all of the inequalities in \eqref{coercive} hold strictly except $c_{21}<0$. 
This is also the case for star molecules. 

The signs of coefficients reflect different modulation mechanism. 
The term $c_{21}|\nabla\bm{p}|^2$ can be stabilized if we truncate up to $M^{(4)}$. 
In fact, we write down the polynomial approximation $\hat{M}^{(4)}$ following the procedure described above, calculate the coefficients by \eqref{coe_ls}, 
and find that the coefficient $c_{41}>0$ for the corresponding term $c_{41}|\nabla^2 \bm{p}|^2$. 
The pair $c_{21}|\nabla \bm{p}|^2+c_{41}|\nabla^2 \bm{p}|^2$ describes the tendency of independent modulation of $\bm{p}$ without coupling to $Q_1$ and $Q_2$. 
We can see this by taking the plane wave $\phi=\exp(i\bm{k}\cdot\bm{x})$ into the energy 
$$
\int\d\bm{x}|\nabla^2\phi|^2+2K|\nabla\phi|^2=(|\bm{k}|^4+2K|\bm{k}|^2)||\phi||^2.
$$
If $K\ge 0$, the preferred frequency is $\bm{k}=0$; 
if $K<0$, the preferred frequency becomes $|\bm{k}|=\sqrt{-K}>0$. 
On the contrary, the quadratic terms about $\nabla Q_1$ and $\nabla Q_2$ 
are positive, indicating that $Q_1$ and $Q_2$ do not tend to show independent modulation, but may show modulation coupled with $\bm{p}$ through the terms $p_i\partial_jQ_{\sigma ij}$. 
Currently, we choose not to include the independent modulation of $\bm{p}$ as an approximation, and discard the term $c_{21}|\nabla\bm{p}|^2$ to avoid lower unboundedness of the energy. 

When we consider molecules with other shapes or interactions and calculate the coefficients from \eqref{coe_ls}, 
we may obtain signs of the coefficients different from the above. 
If this is the case, it indicates that the molecular interaction 
induces different modulation mechanism, 
and we need to do a different truncation in accordance with the mechanism. 

\section{Results and discussion\label{results}}
We examine the phases where inhomogeneity occurs only in the $x$-direction. 
To find modulated phases, we need to minimize the free energy density under the 
periodic boundary condition about the tensors and the period length $L$, 
$$
\min_{\bm{p}(x),Q_1(x),Q_2(x),L}\frac{F[\bm{p}(x),Q_1(x),Q_2(x)]}{L}. 
$$
\subsection{Numerical methods}
We use finite volume method to discretize the free energy. 
Generally speaking, in $[x_k,x_{k+1}]$, a function $g(x)$ is approximated by 
$\frac{1}{2}(g(x_k)+g(x_{k+1}))$, and its derivative is approximated by 
$(g(x_{k+1})-g(x_k))/({x_{k+1}-x_k})$. 
For example, the term 
$$\int_{x_k}^{x_{k+1}}\d x p_i\frac{\d}{\d x}Q_{1,1i}$$ 
is approximated by 
$$
(x_{k+1}-x_{k})\cdot \frac{p_i(x_{k+1})+p_i(x_{k})}{2} \cdot
\frac{Q_{1,1i}(x_{k})-Q_{1,1i}(x_{k+1})}{x_{k+1}-x_k}. 
$$
A single period is discretized using $32$ points. 
The tensors are represented by their eigenvalues and co-owned eigenframe $T(x)$ 
that is represented by the Euler angles $(\alpha(x),\beta(x),\gamma(x))$ 
by (\ref{EulerRep}), 
\begin{align*}
  \bm{p}(x)&=T(x)(s(x),0,0)^T, \\
  Q_1(x)&=T(x)\mbox{diag}(q_{11}(x),q_{12}(x),q_{13}(x))T(x)^T, \\
  Q_2(x)&=T(x)\mbox{diag}(q_{21}(x),q_{22}(x),q_{23}(x))T(x)^T. 
\end{align*}
The eigenvalues are calculated from $(b_1,b_{ij})$ using (\ref{distribution}) and 
(\ref{BRep}). 
We will use $(b_1(x),b_{ij}(x))$ and the Euler angles as the basic variables. 

The derivatives of the free energy about the eigenvalues are given by 
\begin{equation}
\frac{\partial F}{\partial q_{ij}(x)}=b_{ij}(x)+\frac{\partial F_r}{\partial q_{ij}(x)}. \label{GradF}
\end{equation}
Here $F_r$ stands for the part of free energy from the pairwise interaction, 
and the derivatives of the entropy term are calculated by (\ref{average}). 
The derivatives about the Euler angles are given by 
$$
\frac{\partial F}{\partial \alpha(x)}=\frac{\partial F_r}{\partial \alpha(x)}, 
$$
since the entropy term is independent of $T(x)$. 
We use the following stationary point iteration: 
\begin{align}
b_{ij}^{(k+1)}(x)&=b_{ij}^{(k)}(x)-\lambda\frac{\partial F}{\partial q_{ij}^{(k)}(x)}
=(1-\lambda)b_{ij}^{(k)}(x)-\lambda\frac{\partial F_r}{\partial q_{ij}^{(k)}(x)}, \\
\alpha^{(k+1)}(x)&=\alpha^{(k)}(x)-\mu\frac{\partial F_r}{\partial \alpha^{(k)}(x)}. 
\end{align}
The iteration is along a descending direction of the free energy (see \eqref{grad_b}). 

The free energy density may have several local minima. 
Various initial guesses are adopted to obtain as many metastable phases as possible, including but not limited to all the phases presented in the current work. 
Then we compare free energy density of each metastable phase, and label the minimum one as the stable phase. 
In seeking metastable phases, we have also tried with phenomenological coefficients. 
It turns out that many phases can be stable under phenomenological coefficients, but they are found unstable or only metastable under coefficients derived from the hard-core potential. 
In this paper, we only report the stable phases under coefficients derived from the hard-core potential, and leave other metastable phases to a future work. 

\subsection{The phase diagram}
We first list the phases that appear in the phase diagram. 
Define $Q_3=\left<\bm{m}_3\bm{m}_3\right>=I-Q_1-Q_2$ and denote its eigenvalues as $q_{3j}$. 
Because $T$ is the eigenframe shared by $Q_1$ and $Q_2$, it is also the eigenframe of $Q_3$. 
For the phases discussed here, we can do permutation such that $q_{ii}\ge q_{ij}$, and assume this in the following. 
It should be noted that for homogeneous phases, the free energy is independent of the eigenframe $T$. 
\begin{itemize}
\item Isotropic phase ($I$): $s=q_{ij}=0$. 
\item Uniaxial nematic phase ($N_i$): homogeneous with $s=0$, further classified by the relation of eigenvalues. 
In the $N_2$ phase we have 
$q_{22}>1/3>q_{12},q_{32}$ and $q_{j1}=q_{j3}$.
In the $N_3$ phase we have $q_{33}>1/3>q_{13},q_{23}$ and $q_{j1}=q_{j2}$. 
The above relations of eigenvalues indicate that in the $N_i$ phase,
$\bm{m}_i$ aligns near $\pm\bm{n}_i$, and the other two
$\bm{m}_j$ align near the plane perpendicular to $\bm{n}_i$.
\item Biaxial nematic phase ($B$): homogeneous with $q_{ii}>q_{ij}$, 
indicating that $\bm{m}_i$ is preferrably along $\pm\bm{n}_i$. 
\item Twist-bend phase ($N_{tb}$): the eigenvalues $s$ and $q_{ij}$ are constant with $s\ne 0$ and $q_{ii}>q_{ij}$, 
while $T(x)$ shows the modulation 
\begin{equation}
  T(x)=(\bm{n}_1,\bm{n}_2,\bm{n}_3)=\left(
  \begin{array}{ccc}
    0 & -\cos\gamma & \sin\gamma\\
    \cos \frac{\pm 2\pi x}{L} & -\sin\gamma\sin\frac{\pm 2\pi x}{L} 
    & -\cos\gamma\sin\frac{\pm 2\pi x}{L}\\
    \sin \frac{\pm 2\pi x}{L} & \sin\gamma\cos\frac{\pm 2\pi x}{L} 
    & \cos\gamma\cos\frac{\pm 2\pi x}{L}
  \end{array}
  \right), \label{tb}
\end{equation}
where the modulation of $\bm{n}_2$ and $\bm{n}_3$ is identical to the earlier prediction \cite{EuroPhys_56_247}. 
The above equation indicates that $\bm{n}_1$ rotates on a circle, 
and that $\bm{n}_2$ rotates on a conical surface. 
Thus the Euler angle $\gamma$ here becomes the conical angle. 
The sign before $2\pi x$ represents whether $T$ is rotated left- or right-handed. 
The two cases share the same free energy density. 
\end{itemize}
Although we only examine one-dimensional modulated phases, these phases have covered all the phases found experimentally so far. 

\begin{figure}
\centering
\includegraphics[width=.48\textwidth,keepaspectratio]{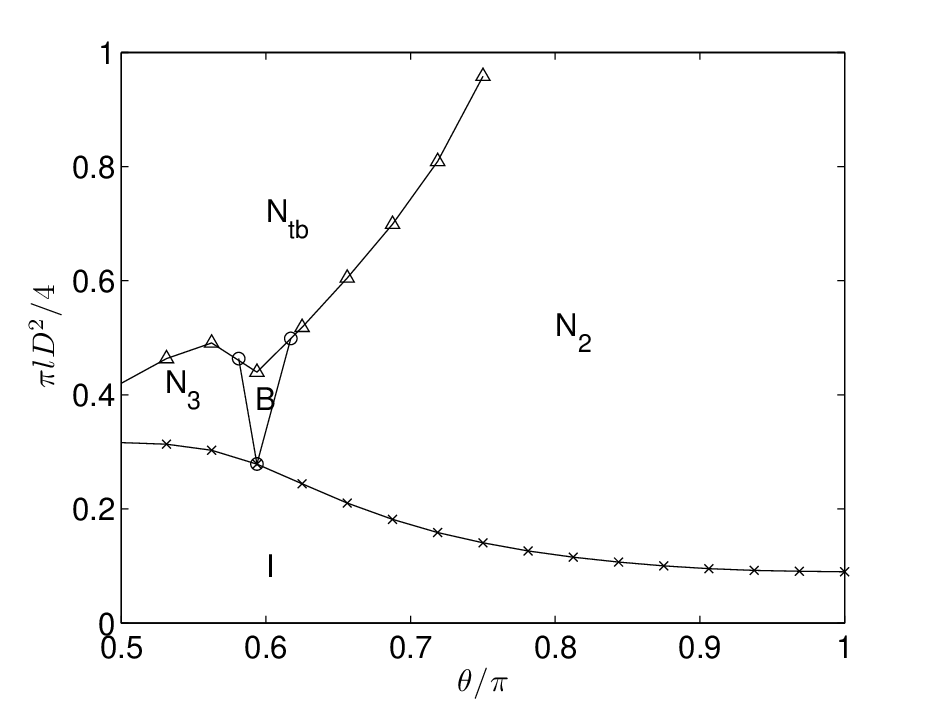}
\includegraphics[width=.48\textwidth,keepaspectratio]{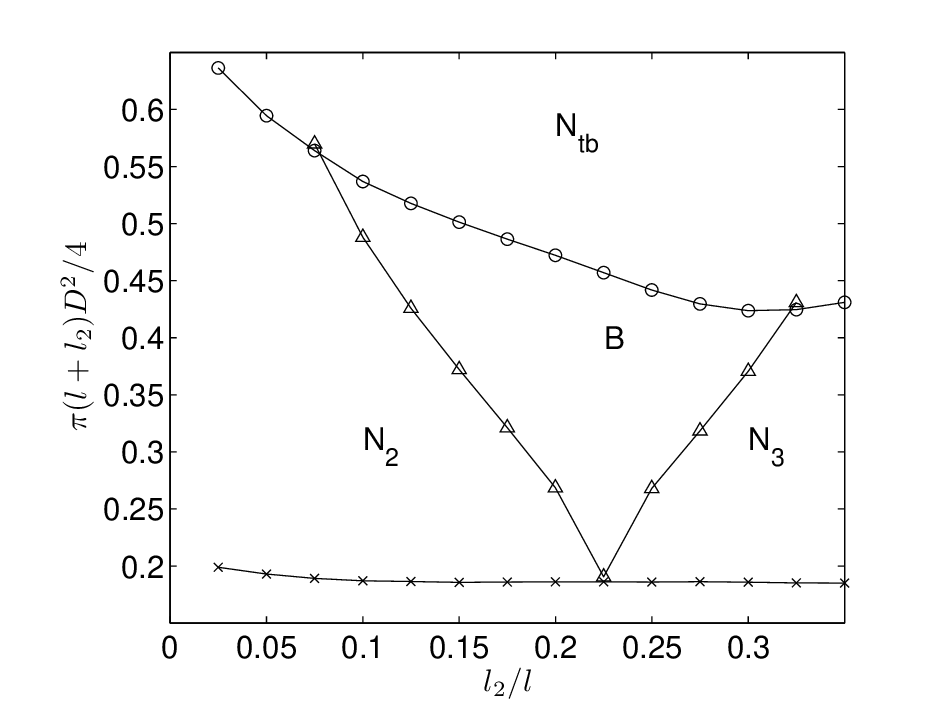}
\caption{\label{phased}Left: phase diagram of bent-core molecules with $\eta=D/l=1/40$. Right: phase diagram of star molecules with $\theta=2\pi/3$, $\eta=1/40$.}
\end{figure}
\begin{figure}
\centering
\includegraphics[width=.48\textwidth,keepaspectratio]{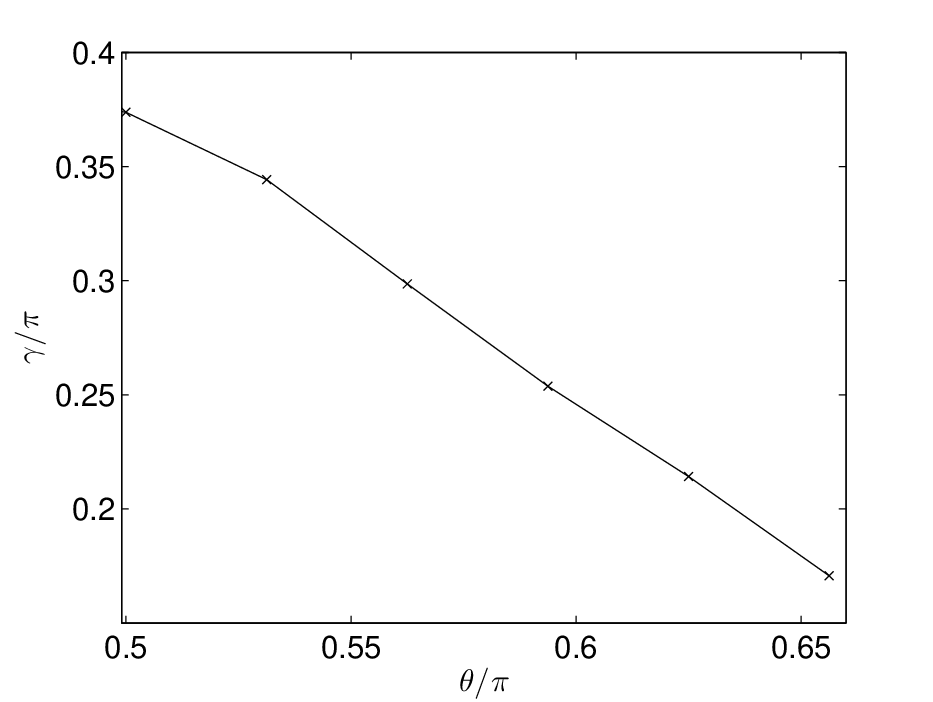}
\includegraphics[width=.48\textwidth,keepaspectratio]{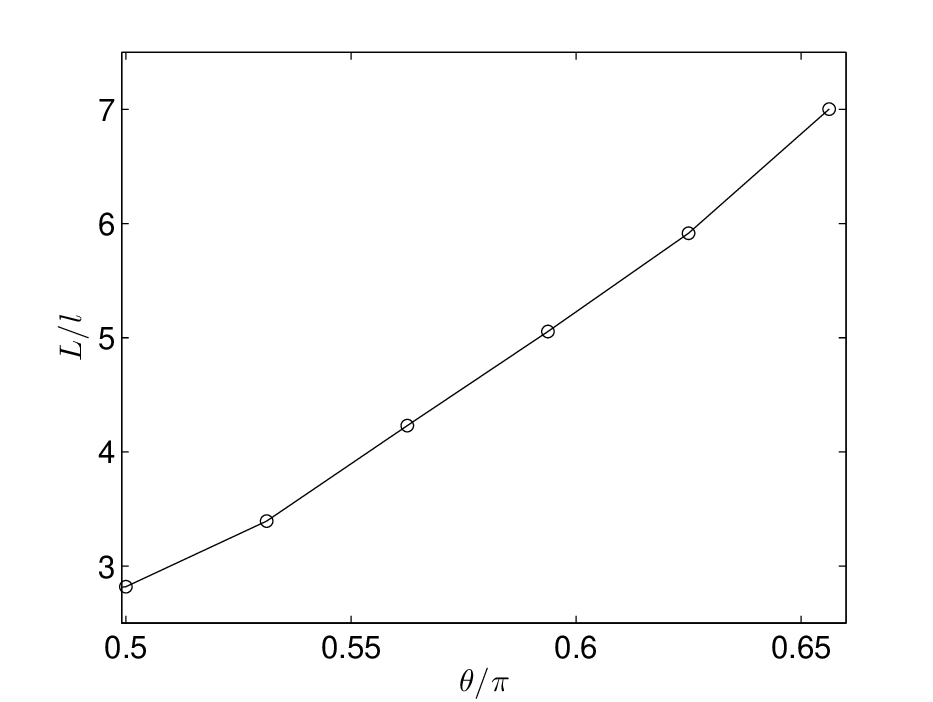}
\caption{\label{pars}The conical angle $\gamma$ and period length $L$ as functions of the bending angle $\theta$ when $\pi lD^2/4=0.7$. }
\end{figure}
The phase diagram of bent-core molecules is given in Fig. \ref{phased} (left), 
where we fix $\eta=1/40$, and use the volume fraction $\pi clD^2/4$ to express the concentration. It shows that $I$ occurs at low volume fraction, and homogeneous nematic phases emerge when it becomes higher. 
As the bending angle $\theta$ decreases from $\pi$, it shows successively $N_2$, $B$ and $N_3$. 
When the volume fraction further grows, the $N_{tb}$ phase occurs if the bending angle $\theta$ is far from $\pi$. 
Experimentally, the $I$--$N_2$--$N_{tb}$ transition is also observed on lowering the temperature for molecules with relatively large bending angle $\theta$ \cite{Ntb,Ntb2}. 
We also plot the conical angle $\gamma$ and period length $L$ as a function of $\theta$ at $\pi clD^2/4=0.7$ (Fig. \ref{pars}). 
We observe that as $\theta$ increases, $\gamma$ decreases while $L$ increases. 
It is worth noting that $L$ is a few times of the dimension of the molecule, giving a very short periodicity that is consistent with the measurements of experiments \cite{Ntb,Ntb2}. 

Next, we study the role of the third arm for star molecules. 
The phase diagram is presented in Fig. \ref{phased} (right). 
Here we fix the bending angle $\theta=2\pi/3$ and focus on the length of the third arm $l_2/l$. 
Now the volume fraction becomes $\pi c(l+l_2)D^2/4$. 
The nematic phases are among those we mentioned above, 
and are sensitive to $l_2$. 
While the transition volume fraction to homogeneous nematic phases is almost unchanged, 
the phase is altered from $N_2$ to $B$ and to $N_3$ when $l_2$ increases. 
The transition volume fraction to $N_{tb}$ is substantially lowered as $l_2$ grows. 
We would view this phase diagram as a typical example of phase behaviors being substantially altered by slight modification on molecular architecture. This is a feature different from rod-like molecules, commonly observed experimentally \cite{JJAP} but not well-understood yet. 

For bent-core molecules, phase diagram about molecular parameters including modulated nematic phases has not been given in existing theoretical models to our knowledge. 
Moreover, in these models that focus on modulated phase, only one director $\bm{n}$, or one second-order tensor $Q$ is included, leading to the absence of the biaxial phase $B$. 
Phase diagram about molecular parameters can only be found in preceding molecular simulations 
\cite{JCP_111_9871,LC_29_483,PRE_67_011703,JCP_123_174907,JCP_129_244903}. 
In these works the molecules studied are thick with $\eta\approx 1/5\sim 1/10$, 
and they did not find the $N_{tb}$ phase. 
The results in \cite{PRL_115_147805} indicate that curved structure can make $N_{tb}$ easier to occur. 
Our results suggest that thin molecules might have the same effect. 

\section{Conclusion\label{concl}}
A tensor model is constructed based on molecular theory for nematic phases of bent-core molecules. 
The free energy is suitable for molecules with the $C_{2v}$ symmetry, 
with the coefficients derived from molecular interaction. 
We use the model to study the nematic phases of bent-core molecules and their analog, star molecules, with the hard-core potential. We obtain the phase diagram about the molecular parameters, including all the nematic phases found experimentally. 

Provided that the molecular symmetry is preserved, the tensor model is able to study 
molecules with arbitrary shape and interactions. 
Hence we aim to apply this model to studying nematic phases of various molecules. 
We are also interested in two- and three-dimensional modulated phases that can be described by the model. 

\appendix
\section{The computation of $M^{(k)}$}
We describe how to compute $M^{(k)}$ for bent-core molecules. 
It works exactly the same way for star molecules. 

\begin{figure}
\centering
\includegraphics[width=.45\textwidth,keepaspectratio]{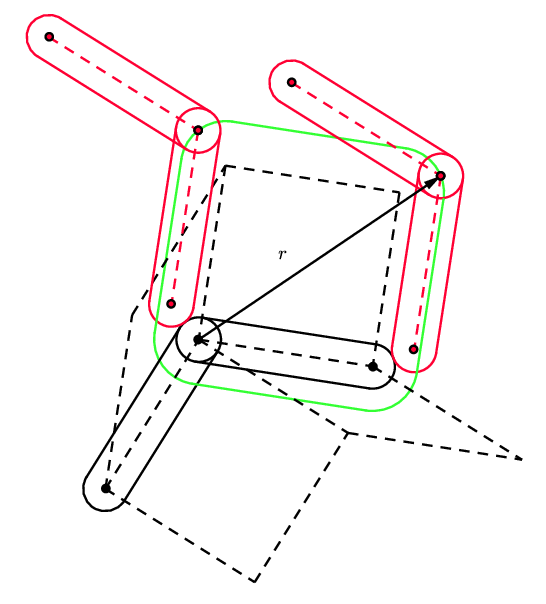}
\hspace{1cm}
\includegraphics[width=.25\textwidth,keepaspectratio]{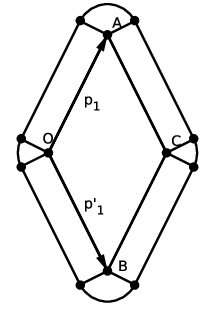}
\caption{\label{region}Left: the region $W$, consisting of four 
spheroparallelograms, whose skeleton parallelograms are drawn in dashed line. 
Right: the intersection of a spheroparallelogram with the plane $z=0$ where the parallelogram lies in. }
\end{figure}
Fix the orientation of a pair of molecules. 
Denote by $W_{ij}$ the region where the relative position let the $i$th arm of one molecule and the $j$th arm of the other touch. 
Then the region where two molecules touch, denoted by $W$, is the union of four $W_{ij}$. 
Each $W_{ij}$ is a spheroparallelogram, obtained by inflating each point in 
a parallelogram to a sphere. One of the $W_{ij}$ is drawn in Fig. \ref{region} (left). 
All the four $W_{ij}$ contain $\hat{O}$, since four arms share the point $\hat{O}$ 
when $\bm{r}=0$. 

Denote
$$
s_{ij}(\bm{n})=\max_{t\bm{n}\in W_{ij}}t. 
$$
Then 
$$
\max_{t\bm{n}\in W}t\triangleq s(\bm{n})=\max_{i,j=1,2}s_{ij}(\bm{n}).
$$ 
For any vector $\bm{n}$, the whole segment $t\bm{n}\ (t\in [0,s(\bm{n})])$ 
lies within $W$, because $W_{ij}$ are convex. 
Hence we can express $M^{(k)}$ by an integral in spherical coordinates, 
which we utilize for numerical calculation, 
\begin{align}
  M^{(k)}(P,P')&=
  \int_{W}\underbrace{\bm{r}\ldots\bm{r}}_{k\ \mbox{\small{times}}}\d\bm{r} 
  =\int_{S^2}\underbrace{\bm{n}\ldots\bm{n}}_{k\ \mbox{\small{times}}}\d\bm{n}
  \int_0^{s(\bm{n})}r^{k+2}\d r\nonumber\\
  &=\int_{S^2}\frac{1}{k+3}s(\bm{n})^{k+3}
  \underbrace{\bm{n}\ldots\bm{n}}_{k\ \mbox{\small{times}}}\d\bm{n}.   
\end{align}

Now it remains to compute $s_{ij}$. 
Place the parallelogram in the plane $z=0$. 
Denote by $R$ the intersection point of the ray $t\bm{n}~(t\ge 0)$ and the boundary of 
the spheroparallelogram. 
The boundary of a spheroparallelogram consists of two planes $z=\pm D$, four cylindrical surfaces at four edges, and four spherical surfaces at four vertices. 
We need to determine where $R$ lies, for which the procedure below is followed: 
\begin{itemize}
\item Compute the intersection point of the ray $t\bm{n}$ and the plane 
$z=\pm D$. Then examine whether its projection on the plane $z=0$ lies in 
the parallelogram $OACB$, drawn in Fig. \ref{region} (right). If it does, the $R$ lies on the flat surface 
of the spheroparallelogram. 
\item Determine whether the ray $t\bm{n}$ intersects with any of the spheres
on the corner. If yes, compute the farthest intersection point and examine 
its projection on the plane $z=0$. If it lies in the corresponding sector 
(located at the corners in Fig. \ref{region}), 
$R$ lies on the spherical surface of the spheroparallelogram. 
\item Now we know that $R$ lies on the cylindrical surface of the 
spheroparallelogram, and it is easy to distinguish which cylinder it locates. 
\end{itemize}

\section{The properties of the Boltzmann distribution}
\subsection{The existence of the Lagrange multiplier}
Let
$$
\mathcal{A}=\{(s,q_{ij})|\rho:SO(3)\to \mathbb{R}^+,~\int\d P \rho=1,~
s=\int \d P \rho m_{11},~q_{ij}=\int\d P \rho m_{ij}^2,~i,j=1,2.\}. 
$$
\begin{theorem}\label{clos_exist}
  Each $(s, q_{ij})\in \mathcal{A}$ is subject to the constraints in 
  \eqref{range}. For any $(s, q_{ij})$ satisfying \eqref{range}, with $s^2<q_{11}$ substituted by $s<q_{11}$, there exists a unique 
  solution to the minimization problem
  \begin{align*}
    &\inf \int_{SO(3)}\d P \rho(P)\log \rho(P),\\
    \mbox{s.t. } &\int \d P \rho(P)=1,\\
    &\int \d P\bm{m}_1\rho(P)=(s,0,0)^T,\\
    &\int \d P\bm{m}_1\bm{m}_1\rho(P)=\mbox{diag}(q_{11},q_{12},1-q_{11}-q_{12}),\\
    &\int \d P\bm{m}_2\bm{m}_2 \rho(P)=\mbox{diag}(q_{21},q_{22},1-q_{21}-q_{22}).
  \end{align*}
  The solution takes the form
  $$
  \rho(P)=\frac{1}{Z}\exp\left(b_1m_{11}+\sum_{i,j=1,2}b_{ij}m_{ij}^2\right), 
  $$
  where
  $$
  Z=\int\d P \exp\left(b_1m_{11}+\sum_{i,j=1,2}b_{ij}m_{ij}^2\right). 
  $$
\end{theorem}
\begin{proof}
  Note that 
  \begin{align*}
    m_{13}^2&=1-m_{11}^2-m_{12}^2\ge 0, \\
    m_{23}^2&=1-m_{21}^2-m_{22}^2\ge 0, \\
    m_{31}^2&=1-m_{11}^2-m_{21}^2\ge 0, \\
    m_{32}^2&=1-m_{12}^2-m_{22}^2\ge 0, \\
    m_{33}^2&=m_{11}^2+m_{12}^2+m_{11}^2+m_{12}^2-1\ge 0, 
  \end{align*}
  and that for $i,j=1,2,3$, the measure of the set $\{P:m_{ij}=0\}$ is zero. 
  Thus the inequalities about only $q_{ij}$ in (\ref{range}) are obtained.
  The inequality about $s$ in (\ref{range}) comes from 
  $(\int\d P fm_{11})^2\le \int\d P fm_{11}^2$, and the equality holds 
  only if $fm_{11}=\lambda f$ holds for a constant $\lambda$, 
  which implies that $f=0$ for $m_{11}\ne\lambda$. Again we note that the 
  measure of the set $\{P:m_{11}=\lambda\}$ is zero. 

  The uniqueness of $f$ is deduced immediately from the strict convexity of 
  $f\log f$ about $f$. 

  To prove the existence, consider the function
  \begin{equation}
  J(b_1,b_{ij})=\int\d P \exp\left(b_1(m_{11}-s)
  +\sum_{i,j=1,2}b_{ij}(m_{ij}^2-q_{ij})\right). \label{JFunc}
  \end{equation}
  A stationary point of $J$ satisfies $\partial J/\partial b_1
  =\partial J/\partial b_{ij}=0$, which yields
  \begin{align*}
    s&=\frac{1}{Z}\int\d P \exp\left(b_1m_{11}
    +\sum_{i,j=1,2}b_{ij}m_{ij}^2\right)m_{11}, \\
    q_{ij}&=\frac{1}{Z}\int\d P \exp\left(b_1m_{11}
    +\sum_{i,j=1,2}b_{ij}m_{ij}^2\right)m_{ij}^2, \quad i,j=1,2. 
  \end{align*}
  Because of the uniqueness, the stationary point of $J$ solves the 
  minimization problem. We will prove that 
  \begin{equation}\label{upperunb}
  \lim_{b_1^2+\sum b_{ij}^2\to \infty}J=+\infty. 
  \end{equation}
  Since $J$ is bounded from below, (\ref{upperunb}) indicates the existence 
  of a minimizer. 
  
  For (\ref{upperunb}), it is sufficient to prove that for any 
  $$
  (b_1,b_{11},b_{12},b_{21},b_{22})\neq (0,0,0,0,0), 
  $$
  there exists a $P$ such that 
  $$
  I(P)=b_1(m_{11}-s)+\sum_{i,j=1,2}b_{ij}(m_{ij}^2-q_{ij})>0. 
  $$
  Let
  \begin{align*}
    &P_1=\left(
    \begin{array}{ccc}
      1&0&0\\
      0&0&1\\
      0&-1&0
    \end{array}
    \right),
    &P_2=&\left(
    \begin{array}{ccc}
      -1&0&0\\
      0&0&1\\
      0&1&0
    \end{array}
    \right),
    &P_3=&\left(
    \begin{array}{ccc}
      1&0&0\\
      0&1&0\\
      0&0&1
    \end{array}
    \right),
    &P_4=&\left(
    \begin{array}{ccc}
      -1&0&0\\
      0&1&0\\
      0&0&-1
    \end{array}
    \right),\\
    &P_5=\left(
    \begin{array}{ccc}
      0&1&0\\
      0&0&1\\
      1&0&0
    \end{array}
    \right),
    &P_6=&\left(
    \begin{array}{ccc}
      0&0&1\\
      1&0&0\\
      0&1&0
    \end{array}
    \right),
    &P_7=&\left(
    \begin{array}{ccc}
      0&0&-1\\
      0&1&0\\
      1&0&0
    \end{array}
    \right),
    &P_8=&\left(
    \begin{array}{ccc}
      0&1&0\\
      1&0&0\\
      0&0&-1
    \end{array}
    \right).
  \end{align*}
  It is straightforward to verify that 
  $$
  \lambda_iI(P_i)=0
  $$
  holds for arbitrary $(b_1,b_{ij})$ where 
  \begin{align*}
    \lambda_1&=(\lambda-q_{22})\left(\frac{1}{2}+\frac{s}{2q_{11}}\right),&
    \lambda_2&=(\lambda-q_{22})\left(\frac{1}{2}-\frac{s}{2q_{11}}\right),\\
    \lambda_3&=(q_{11}+q_{22}-\lambda)\left(\frac{1}{2}+\frac{s}{2q_{11}}\right),&
    \lambda_4&=(q_{11}+q_{22}-\lambda)\left(\frac{1}{2}-\frac{s}{2q_{11}}\right),\\
    \lambda_5&=1-\lambda-q_{12},&
    \lambda_6&=1-\lambda-q_{21},\\
    \lambda_7&=\lambda-q_{11},&
    \lambda_8&=q_{12}+q_{21}-(1-\lambda).
  \end{align*}
  Here $\lambda$ is a real number to be determined. We choose a $\lambda$ 
  such that $\lambda_i> 0$. It is equivalent to
  \begin{align*}
    \lambda -q_{22},\ q_{11}+q_{22}-\lambda,\ 
    \lambda -q_{11},\
    1-\lambda -q_{12},\ 
    1-\lambda -q_{21},\ 
    q_{21}+q_{12}-(1-\lambda )>0,
  \end{align*}
  which yields
  $$
  \max\{q_{11},q_{22},1-q_{12}-q_{21}\}<\lambda<\min\{1-q_{12},1-q_{21},q_{11}+q_{22}\}.
  $$
  From the constraints on $q_{ij}$, the upper bound is greater than the lower 
  bound, which guarantees the existence of $\lambda$. Note that 
  $$
  \sum_{i=1}^8\lambda_{i}=1. 
  $$
  Let
  $$
  A=b_1s+\sum_{i,j=1,2}b_{ij}q_{ij}.
  $$
  We claim that $I(P_i)>0$ for some $i$. Otherwise $I(P_i)=0$ for every $i$. 
  Expanding these equalities, we have
  \begin{align*}
    \pm b_1+b_{11}=\pm b_1+b_{11}+b_{22}=b_{21}=b_{12}=b_{22}=b_{12}+b_{21}=A.
  \end{align*}
  It is easy to deduce that $b_1=b_{ij}=0$. 
\end{proof}

\subsection{Some equalities}
The derivatives of $F_{entropy}$ about the tensors are 
\begin{equation}\label{average}
  \frac{1}{\beta_0}\frac{\partial F_{entropy}}{\partial (\bm{p},Q_1,Q_2)}=(\bm{b},B_1,B_2). 
\end{equation}
We prove it for $\bm{p}$ as an example. Note that 
\begin{equation}
  \frac{\partial\log Z}{\partial (\bm{b},B_1,B_2)}=
  \frac{1}{Z}\frac{\partial Z}{\partial (\bm{b},B_1,B_2)}=(\bm{p},Q_1,Q_2). 
\end{equation}
Hence 
\begin{align*}
  &\frac{1}{\beta_0}\frac{\partial F_{entropy}}{\partial \bm{p}}\\
  =&\frac{\partial (\bm{b}\cdot\bm{p}+B_1:Q_1+B_2:Q_2-\log Z)}
  {\partial \bm{p}}\\
  =&\bm{b}
  +\bm{p}\cdot\frac{\partial \bm{b}}{\partial \bm{p}}
  +Q_1:\frac{\partial B_1}{\partial \bm{p}}
  +Q_2:\frac{\partial B_2}{\partial \bm{p}}
  -\frac{\partial \log Z}{\partial \bm{p}}\\
  =&\bm{b}
  +\frac{\partial \log Z}{\partial \bm{b}}\cdot\frac{\partial \bm{b}}{\partial \bm{p}}
  +\frac{\partial \log Z}{\partial B_1}:\frac{\partial B_1}{\partial \bm{p}}
  +\frac{\partial \log Z}{\partial B_2}:\frac{\partial B_2}{\partial \bm{p}}
  -\frac{\partial \log Z}{\partial \bm{p}}\\
  =&\bm{b}+\frac{\partial \log Z}{\partial (\bm{b},B_1,B_2)}\cdot\frac{\partial (\bm{b},B_1,B_2)}{\partial \bm{p}}-\frac{\partial \log Z}{\partial \bm{p}}\\
  =&\bm{b}. 
\end{align*}

The derivatives of $F$ about $b_{ij}(x)$ can be written as 
\begin{equation}
\frac{\partial F}{\partial b_{ij}(x)}
=\frac{\partial q_{kl}(x)}{\partial b_{ij}(x)}\frac{\partial F}{\partial q_{kl}(x)}. 
\end{equation}
And note that 
$$
\frac{\partial q_{kl}(x)}{\partial b_{ij}(x)}=\frac{\partial^2 Z}{\partial b_{ij}(x)\partial b_{kl}(x)} 
=\left<(m_{11},m_{11}^2,m_{12}^2,m_{21}^2,m_{22}^2)^T(m_{11},m_{11}^2,m_{12}^2,m_{21}^2,m_{22}^2)\right>. 
$$
is positive definite. Thus 
\begin{equation}
\left(\frac{\partial F}{\partial b_{ij}(x)}\right)^T\frac{\partial F}{\partial q_{ij}(x)}=\left(\frac{\partial F}{\partial q_{ij}(x)}\right)^T\frac{\partial^2 Z}{\partial b_{ij}(x)\partial b_{kl}(x)} \frac{\partial F}{\partial q_{ij}(x)}>0. 
\label{grad_b}
\end{equation}

\textbf{Acknowledgments}. 
P. Zhang is partly supported by National Natural Science Foundation of China 
(Grant No. 11421101 and No. 11421110001). 

\bibliographystyle{plain}
\bibliography{bib_bent} 

\end{document}